\documentclass[letterpaper, 10 pt, conference]{ieeeconf}  

\IEEEoverridecommandlockouts                              
\usepackage{cite}
\usepackage{amsmath,amssymb,amsfonts,bm}
\usepackage{bbm}
\usepackage{graphicx}
\usepackage{subfigure}
\usepackage{textcomp}
\usepackage{algorithm}
\usepackage{algpseudocode}

\usepackage{xcolor}
\allowdisplaybreaks[4]

\newtheorem{theorem}{Theorem}
\newtheorem{lemma}{Lemma}

\newcommand{\mathletter}[1]{%
	\expandafter\newcommand\csname b#1\endcsname{\mathbb #1}
	\expandafter\newcommand\csname c#1\endcsname{\mathcal #1}
	\expandafter\newcommand\csname f#1\endcsname{\mathfrak #1}
	\expandafter\newcommand\csname til#1\endcsname{\widetilde #1}
	\expandafter\newcommand\csname ha#1\endcsname{\widehat #1}
	\expandafter\newcommand\csname s#1\endcsname{\boldsymbol #1}
}%
\def\mathletters#1{\mathlettersB #1,,}
\def\mathlettersB#1,{\ifx,#1,\else\mathletter #1\expandafter\mathlettersB\fi}
\mathletters{A,B,C,D,E,F,G,H,I,J,K,L,M,N,O,P,Q,R,S,T,U,V,W,X,Y,Z}
\def\tr{\text{tr}}

\def\ox{\overline{X}}
\def\ou{\overline{U}}
\def\ha{\widehat{A}}
\def\hb{\widehat{B}}

\def\bee{\begin{equation}}
	\def\ene{\end{equation}}
\def\been{\begin{equation*}}
	\def\enen{\end{equation*}}

\title{\LARGE \bf
Linear Convergence of Data-Enabled Policy Optimization for Linear Quadratic Tracking
}

\author{Shubo Kang, Feiran Zhao, Keyou You
\thanks{This work was supported by National Science and Technology Major Project of China (2022ZD0116700) and National Natural Science Foundation of China (62033006, 62325305). (Corresponding author: Keyou You)}
\thanks{S. Kang, and K. You are with the Department of Automation and BNRist, Tsinghua University, Beijing 100084, China. (e-mail: ksb20@mails.tsinghua.edu.cn, youky@tsinghua.edu.cn.) }
\thanks{F. Zhao is with the Department of Information Technology and Electrical Engineering, ETH Z\"{u}rich, 8092 Z\"{u}rich, Switzerland. (e-mail: zhaofe@control.ee.ethz.ch)}}%
\begin{document}

\maketitle
\thispagestyle{empty}
\pagestyle{empty}

\begin{abstract}
	Data-enabled policy optimization (DeePO) is a newly proposed method to attack the open problem of direct adaptive LQR. In this work, we extend the DeePO framework to the linear quadratic tracking (LQT) with offline data. By introducing a covariance parameterization of the LQT policy, we derive a direct data-driven formulation of the LQT problem. Then, we use gradient descent method to iteratively update the parameterized policy to find an optimal LQT policy. Moreover, by revealing the connection between DeePO and model-based policy optimization, we prove the linear convergence of the DeePO iteration. Finally, a numerical experiment is given to validate the convergence results. We hope our work paves the way to direct adaptive LQT with online closed-loop data. 
\end{abstract}


\section{Introduction}
Over the last decade, data-driven approaches have emerged as a major development in control \cite{annaswamy2024control}. The various data-driven control methods generally be categorized into two primary classes: \emph{indirect} methods, which involve system identification followed by model-based control design; \emph{direct} methods, which bypass the model estimation step and design the controller directly from data. As an end-to-end method, direct data-driven control are conceptually straightforward and relatively easy to implement in practice. Consequently, numerous direct methods and theories have been proposed in recent years, such as Data-EnablEd Predictive Control (DeePC) \cite{coulson2019data,chiuso2023harnessing}, data informativity \cite{van2020data, van2020noisy}, and direct Linear Quadratic Regulator (LQR) methods \cite{dorfler2023certainty, de2021low}, among others.

A central problem in direct data-driven control is the \emph{adaptive} designs that can update control policies using real-time closed-loop data. The problem remained unsolved until recently, when our prior work \cite{zhao2023data,zhao2024data} introduced a novel approach named \textbf{Data-EnablEd Policy Optimization (DeePO)}. DeePO draws on the principles of policy optimization (PO), which is well-known by its application in reinforcement learning \cite{bertsekas2019reinforcement}. The theoretical analysis for data-driven PO used to base on the zeroth-order optimization, and has given convergence results  forstabilizing control \cite{zhao2024convergence}, LQR \cite{fazel2018global,mohammadi2021convergence}, , mix $\cH_2$-$\cH_\infty$ control \cite{zhang2019policy}, etc. Although the PO framework is well-suited for adaptive control, zeroth-order optimization is not since it needs a large number of complete trajectory data to obtain the optimal policy, making it impossible to be implemented online. In contrast, our DeePO method calculates the policy gradient using only a batch of persistently exciting (PE) data. Further, by introducing a new data-based policy parameterization, the policy can be updated efficiently using the online data. In \cite{zhao2024data}, the online DeePO algorithm is proved to converge at a sublinear rate $\cO(1/\sqrt{T})$, comparable to the best indirect adaptive control algorithms. 

Previous studies on DeePO have primarily focused on the LQR problem. As a natural extension, the linear quadratic tracking (LQT) problem aims to track a reference trajectory by simultaneously minimizing the tracking error and energy consumption. Unlike LQR, the optimal LQT policy contains an additional feedforward term dependent on the reference signal \cite{lewis2012optimal}, making the policy parameterization in \cite{zhao2024data} unsuitable for direct application to LQT.  Hence, extending the DeePO approach to the LQT problem is nontrivial and warrants further investigation. 

In this work, we extend the offline DeePO to solve the LQT problem. Our contributions are as following:
\begin{itemize}
	\item We propose a covariance parameterization that fits the LQT policy, extending the DeePO parameterization in \cite{zhao2024data}. Then, a data-driven PO method is
	derived, which achieves performance comparable to indirect methods without requiring system identification. Additionally, since the theory of DeePO relies heavily on its offline analysis like the gradient domination, the method serves as the vital foundation of the future research on online adaptive LQT.
	\item We establish the first linear convergence proof for the offline DeePO method by revealing its connection to model-based policy optimization, thereby improving the sublinear rate established in \cite{zhao2024data}. A numerical experiment is implemented to validate the linear convergence. 
\end{itemize}

The rest of this paper is organized as follows. In Section II, we recapitulate the model-based and data-driven LQT problem. In Section III, we derive the DeePO for LQT and prove its convergence property. In Section IV, a numerical experiment is implemented to support our result. Conclusion and future work in Section V complete this paper.

\textbf{Notation.} We use $I_n$ to denote the $n$ by $n$ identify matrix. We use $\Vert \cdot \Vert$ to denote the $2$-norm and $\rho(\cdot)$ to denote the spectral radius of a square matrix. We use $\underline{\sigma}(\cdot)$ to denote the minimal singular value of a matrix. We use $\dagger$ to denote the right inverse of a full row rank matrix.

\section{Data-driven formulation of the linear quadratic tracking}
In this section, we revisit the model-based and indirect data-driven LQT problem. Subsequently, we present the problem that will be addressed in this work.

\subsection{The model-based LQT problem}
Consider the linear time invariant system:
\bee \label{linearsys}
x_{t+1} = Ax_t +Bu_t + w_t,
\ene
where $x_t \in \bR^n$ and $u_t \in \bR^m$ are the system states and inputs, $w_t\sim\cN(0,W)$ is i.i.d Gaussian noise. We assume the system $(A, B)$ is controllable. 

In this work, we consider the LQT problem of the setpoint case, which aims to regulate the system to track a constant signal $z_t = \delta, t = 0,1,\cdots$. The LQ cost at each step is designed to be $(x_t-\delta)^{\top}Q(x_t-\delta)+u_t^{\top}Ru_t$, where $Q \succ 0$ and $R \succ 0$ are the penalty matrices. Since the one-step LQ cost may not converge to zero when $\delta \neq 0$, the LQT is phrased as finding a policy $\theta = [K~l]$ to minimize the average cost:
\bee \label{cost}
\begin{aligned}
&J(\theta) = \\
&~~~~\limsup_{T\rightarrow\infty}~\frac{1}{T} \bE\left[\sum_{t=0}^{T-1}(x_t-\delta)^{\top}Q(x_t-\delta)+u_t^{\top}Ru_t\right].
\end{aligned}
\ene
where $u_t = K x_t + l$ starting from the initial state $x_0$. 

If the system model $(A,B)$ is known, the unique optimal policy is well-established \cite{lewis2012optimal}:
\bee
\begin{aligned} \label{model_based_solution}
	K^* &= -(B^\top P^* B + R)^{-1}B^\top P^* A, \\
	l^* &= -(B^\top P^* B + R)^{-1}B^\top(I - A - BK)^{-\top} Q\delta,
\end{aligned}
\ene
and $P^*$ is the unique positive semi-definite solution to the discrete-time algebraic Riccati equation:
\bee \label{model_based_solution_2}
	P^* = A^\top P^* A + Q + A^\top P^*B(B^\top P^* B + R)^{-1} B^\top P^* A.
\ene

\subsection{Indirect data-driven formulation for LQT}
In this work, we assume that the system $(A,B)$ are unknown. Instead, we have access to an offline dataset consisting of $T$-length sequences of states and inputs:
$$
\begin{aligned}
	X_{0} &:= \begin{bmatrix}
		x_0& x_1& \dots& x_{T-1}
	\end{bmatrix}\in \mathbb{R}^{n\times T},\\
	U_{0} &:= \begin{bmatrix}
		u_0& u_1& \dots& u_{T-1}
	\end{bmatrix}\in \mathbb{R}^{m\times T}, \\
	W_{0} &:= \begin{bmatrix}
		w_0& w_1& \dots& w_{T-1}
	\end{bmatrix}\in \mathbb{R}^{n\times T}, \\
	X_{1} &:= \begin{bmatrix}
		x_1& x_2& \dots& x_T
	\end{bmatrix}\in \mathbb{R}^{n\times T},
\end{aligned}
$$
which satisfies the system dynamics \eqref{linearsys} as
$$
X_{1} = A X_{0} + B U_{0} + W_{0}.
$$
We assume the data is persistently exciting \cite{willems2005note}, i.e., $D_0 \triangleq [U_0^\top~~X_0^\top]^\top$ has full row rank. The assumption is standard and widely adopted in the data-driven control setting \cite{de2019formulas, kang2023minimum}.

The indirect data-driven method first estimates a model $(\ha,\hb)$ from the data and subsequently applies any model-based control method. Define the least square estimate of $[A,B]$ as $[\ha~\hb]= \mathop{\arg\min}_{A,B}\Vert X_1 - AX_0 - BU_0\Vert$. Then, the indirect tracking problem is
\bee \label{ce_problem}
\begin{aligned}
	&\min \limits_{\theta}~ J(\theta), \\
	&~s.t. ~~x_{t+1} = \ha x_t +\hb u_t + w_t, \\
	&~~~~~~~~~~u_t = K x_t + l.
\end{aligned}
\ene
We can replace $A,B$ by $\ha,\hb$ in \eqref{model_based_solution} and \eqref{model_based_solution_2} to obtain the unique solution of \eqref{ce_problem}, which we denote as $\hat{K}$ and $\hat{l}$.

\subsection{Policy optimization approach for the LQT}

The LQT problem \eqref{ce_problem} can also be solved via policy optimization \cite{zhao2023global}. Specifically, let $\nabla_\theta J(\theta)$ be the policy gradient, the iterative update
\bee
	\theta^+ = \theta - \eta \nabla_\theta J(\theta)
\ene
converges linearly with a small enough stepsize $\eta$. 

However, calculating the policy gradient on $K$ and $l$ requires the system $A,B$ or their estimate $\ha,\hb$. In this work, we aim to propose a direct data-driven approach to solve the LQT problem via policy optimization,wherein the policy gradient is computed directly from the raw data matrices, without requiring explicit model identification.

\section{Data-enabled policy optimization for LQT} \label{sec::main}
In this section, we first propose a covariance parameterization for the LQT policy, enabling us to formulate the LQT problem in a direct data-driven fashion. Building on this, we propose our DeePO method to solve the LQT problem. By revealing the connection between DeePO and traditional model-based policy optimization, we establish the linear global convergence of our approach.

\subsection{Covariance parameterization for the LQT policy}
In previous work, DeePO \cite{zhao2024data} introduced a novel parameterization for the LQR problem based on the data covariance, which allows the policy gradient to be computed directly from data. However, for the LQT problem, the presence of a feedforward term $l$ necessitates modifying the existing covariance parameterization to accommodate tracking control. 

First, define data sample covariance as
$$
\Lambda : = \frac{1}{T}\begin{bmatrix}
	U_0 \\
	X_0
\end{bmatrix}\begin{bmatrix}
	U_0 \\
	X_0
\end{bmatrix}^{\top},
$$
Also, define the following verified data matrices:
$$
\begin{aligned}
	&\ox_0= \frac{1}{T}X_0D_0^{\top}, \ou_0= \frac{1}{T}U_0D_0^{\top},\\
	&\overline{W}_0= \frac{1}{T}W_0D_0^{\top}, \ox_1= \frac{1}{T}\ox_1D_0^{\top}.
\end{aligned}
$$
Then, the policy parameter $K,l$ can be reparameterized as $$\xi = [V~h]$$ satisfying 
\bee \label{parametrization}
\begin{bmatrix} K \\ I_n \end{bmatrix}=\Lambda V = \begin{bmatrix}\ou_0 \\ \ox_0\end{bmatrix}V, 
\begin{bmatrix} l \\ 0 \end{bmatrix}=\Lambda h=\begin{bmatrix}\ou_0 \\ \ox_0\end{bmatrix}h.
\ene
By \eqref{parametrization}, the system dynamics can be represented via data matrices and the policy $\xi$ as
\bee \label{close_loop_1}
A+BK = (B\ou_0+A\ox_0)V = \ox_1V - \overline{W}_0V,
\ene
and
\bee \label{close_loop_2}
Bl = (B\ou_0+A\ox_0)h = \ox_1 h - \overline{W}_0h.
\ene

Since the noise term $W_0$ is unknown, we disregard $W_0$ in \eqref{close_loop_1} and \eqref{close_loop_2} based on the certainty equivalence principle \cite{dorfler2023certainty}. Thus, we formulate the direct data-driven LQT problem as the following data-based optimization:
\bee \label{data_driven_problem}
\begin{aligned}
	&\min \limits_{\xi}~ J(\xi), \\
	&~s.t. ~~x_{t+1} = \ox_1V x_t +\ox_1 h + w_t, \\
	&~~~~~~~~~~u_t = \ou_0 V x_t + \ou_0 h.
\end{aligned}
\ene
where, with a slight abuse of notation, $J(\xi)$ is the cost function under policy $\xi$. $J(\xi)$ is defined to equal to $J(\theta)$ as long as $\xi$ and $\theta$ satisfy \eqref{parametrization}. The problem \eqref{data_driven_problem} can be seen as an approximation of minimizing \eqref{cost}.

Compared to the conventional policy parameterization in \cite{de2019formulas}, the covariance parameterization \eqref{parametrization} offers several benefits. First, $\xi$ have constant size, independent of the data length $T$. Second, there is a unique mapping between $\xi$ and the model-based parameter $\theta$, ensuring consistency between data-driven and model-based approaches. Third, the parameterization has an intrinsic connection with the least-squares estimate of the system matrices $(\ha,\hb)$. To illustrate this connection, note that
\bee \label{conn_data_estimation}
\begin{aligned}
	&\ha\ox_0 + \hb \ou_0 = [\hb~\ha]\begin{bmatrix}
		\ou_0 \\ \ox_0
	\end{bmatrix} \\
	&= \left(X_1 D_0^T (D_0 D_0^T)^{-1}\right)\left(\frac{1}{T}D_0 D_0^\top\right) = \ox_1.
\end{aligned}
\ene
Furthermore, it holds that
\begin{subequations} \label{conn_data_estimation_2}
	\bee \ha + \hb K = \left(\hb\ou_0+\ha\ox_0\right)V = \ox_1 V, \ene
	\bee \hb l = \left(\hb\ou_0+\ha\ox_0\right)h = \ox_1 h.\ene
\end{subequations}
The equation \eqref{conn_data_estimation} and \eqref{conn_data_estimation_2} show the equivalence of the LQT problem \eqref{data_driven_problem} and \eqref{ce_problem}, which is utilized to prove the linear convergence in this work. These advantages are also proved to be significant to design an adaptive algorithm \cite{zhao2024data}, which will be extended to LQT in the future work.

\subsection{The DeePO algorithm for solve \eqref{data_driven_problem}} 
We will show the policy gradient of \eqref{data_driven_problem}, and then give our DeePO algorithm based on the gradient. 

Hereafter, we define the feasible policy set by
$$\cS = \left\{\xi|\ox_0V = I, \ox_0h = 0 \text{ and } \rho(\ox_1 V)<1 \right\}.$$
Then for some $\xi \in \cS$, define $P_V$ as the unique positive semi-definite solution of the Lyapunov equation:
$$
	P_V = Q + V^\top \ou_0^\top R \ou_0V+V^\top \ox_1^\top P_V \ox_1V.
$$
Besides, define the following policy-dependent values:
\begin{gather}
	Y_V = (I - \ox_1V)^{-1}, E_V = (\ou_0^\top R\ou_0 + \ox_1^\top P_V\ox_1)V, \nonumber \\
	g^\top_\xi = (-\delta^\top Q + h^\top \ou_0^\top R\ou_0V + h^\top \ox_1^\top P_V \ox_1 V)Y_V, \nonumber \\
	G_\xi = \ox_1^\top g_\xi + \ou_0^\top R\ou_0 h\ + \ox_1^\top P_V\ox_1 h, \nonumber \\
	\Sigma_V = W + \ox_1 V \Sigma_V V^\top \ox_1^\top. \nonumber
\end{gather}

We give the explicit form of $J(\xi)$ and the policy gradient $\nabla_\xi J(\xi)$ in the following lemmas. 
\begin{lemma}[Cost function] \label{theo_cost}
	For any $\xi \in \cS$, it follows that 
	$$
	\begin{aligned}
		J(\xi) = &\ \delta^\top Q \delta + h^\top \ou_0^\top R\ou_0 h + h^\top \ox_1^\top P_V\ox_1 h\\
		&+ 2g^\top_\xi \ox_1h+tr(P_VW).
	\end{aligned}
	$$
\end{lemma}
\begin{lemma}[Policy gradient] \label{theo_gradient}
	For any $\xi\in\cS$, the gradient of $J(\xi)$ in $\xi$ is 
	$$\nabla_\xi J(\xi) = 2\begin{bmatrix} E_V & G_\xi \end{bmatrix} \Phi_{\xi},$$
	where $\overline{x} = Y_V\ox_1h$ is the mean of the stationary state, and
	$$\Phi_{\xi} = \begin{bmatrix} \Sigma_V + \overline{x}\overline{x}^\top & \overline{x} \\ \overline{x}^\top & 1\end{bmatrix}.$$
\end{lemma}

The proofs are given in the Appendix. The feasible set $\cS$ has constrains $\ox_0V = I$ and $\ox_0h = 0$. Therefore, to ensure feasibility, we need to project the policy gradient $\nabla_\xi J(\xi)$ onto the null space of $\ox_0$. To achieve this, we employ the following projected gradient descent update rule for our DeePO algorithm:
\bee \label{projected_grad}
\xi^+ = \xi - \eta \Pi_{\ox_0} \nabla_\xi J(\xi),
\ene
where $\eta > 0$ is the stepsize, $\Pi_{\ox_0} := I - \ox_0^\dagger \ox_0$ is the projection operator, and the superscript $(\cdot)^+$ indicates the updated policy parameters after one iteration. Since the policy gradient $\nabla_\xi J(\xi)$ can be directly calculated using the data, this method is considered a direct data-driven method.

The feasible set $\cS$ also includes the constraint $\rho(\ox_1 V)<1$, which ensures the stability of the close loop system. However, we do not enforce this constraint directly through the projection step. Instead, we show later that, with a suitably chosen step size $\eta$, the stability constraint is satisfied automatically as the iterations progress.

\subsection{Global linear convergence of the DeePO algorithm} \label{sec::converge}
Due to the non-convexity of both $J(\xi)$ and $\cS$, the convergence of DeePO iteration \eqref{projected_grad} is not trivial. However, by translating \eqref{projected_grad} into a scaled model-based gradient descent iteration, we can leverage established results from model-based policy optimization \cite{zhao2023global} to demonstrate the global linear convergence of \eqref{projected_grad}. 

In the following, we formally establish the connection between our DeePO method with the model-based one.

\begin{lemma} \label{theo_equiv}
	For any $\xi \in \cS$,  the DeePO iteration \eqref{projected_grad} is equivalent to
	\bee \label{equ_model_based_po}
	\theta^+ = \theta - \eta M \nabla_\theta J(\theta),
	\ene
	where $ M = \ou_0(I-\ox_0^\dagger \ox_0)\ou_0^T$ depends only on the offline data. Furthermore, $M$ is positive definite.
\end{lemma}

The proof is given in the Appendix. By Lemma \ref{theo_equiv}, we only need to prove the convergence of \eqref{equ_model_based_po}, which is written under the model-based framework. Similar to the well-known work \cite{fazel2018global}, the proof is completed by providing the Lipschitz smoothness and gradient domination of the cost function $J(\theta)$ on the feasible set.

In the following, define sets
$$
\begin{aligned}
\cS_\theta &= \left\{\theta|\xi \in \cS \text{ and } \xi,\theta \text{ satisfy \eqref{parametrization}}\right\}, \\
\cS_\theta(a) &= \left\{\theta|\xi \in \cS, J(\xi) \leq a \text{ and } \xi,\theta \text{ satisfy \eqref{parametrization}}\right\}.
\end{aligned}
$$

\begin{lemma}[Gradient domination, {\cite[Lemma 6]{zhao2023global}}] \label{grad_dom}
	For any $a \geq J^*$ and $\theta \in \cS_\theta(a)$, there exists $\lambda(a) > 0$ such that
	\bee \label{gradient domination}
	J(\theta)-J^* \leq \lambda(a) \Vert\nabla_\theta J(\theta) \Vert^2,
	\ene
	where $J^*$ is the optimal LQT cost of \eqref{ce_problem}.
\end{lemma}

Also, we need to show the local Lipschitz smoothness of $J(\theta)$.

\begin{lemma}[Local smoothness] \label{lip_smooth}
	For any $a \geq J^*$ and for any $\theta, \theta' \in \cS_\theta(a)$ such that $\theta + \phi(\theta'-\theta) \in \cS_\theta(a)$ for all $\phi \in (0,1)$, there exists $\ell(a) > 0$ such that
	\bee
	J(\theta) \leq J(\theta') + \langle \nabla_\theta J(\theta), \theta' - \theta \rangle + \ell(a)\Vert \theta'-\theta\Vert^2/2.
	\ene
\end{lemma}

\begin{proof}
	The local smoothness of $J(\theta)$ can be proved by providing an upper bound for $\Vert \nabla^2 J(\theta)\Vert := \sup_{\Vert Z\Vert_F = 1}~\nabla^2J(\theta)[Z,Z]$. The complete proof is tedious and hence omitted due to limited space. 
\end{proof}

By Lemma \ref{grad_dom} and \ref{lip_smooth}, it is able to prove the global linear convergence of the model-based iteration \eqref{equ_model_based_po}. For a initial policy $\theta^0\in\cS_\theta$, define $l^0 = l(J(\theta^0))$ and $\gamma^0 = \gamma(J(\theta^0))$. Then we have the following convergence result.

\begin{lemma} \label{theo_conv_theta}
	For $\theta^0 \in \cS_\theta$ and stepsize $\eta \in (0,1/(\ell^0 \Vert M \Vert)]$, the update \eqref{equ_model_based_po} leads to $\theta^{t} \in \cS_\theta(J(\theta^0)), \forall t \in \bN$. Moreover, $\theta^{t}$ satisfies
	$$
	J(\theta^{t+1}) - J^* \leq \left(1 -\nu\right)(J(\theta^t)-J^*),
	$$
	where $\nu = \eta\left(1-\ell^0\Vert M \Vert\eta /2 \right)\underline{\sigma}(M)/ \lambda^0$. 
\end{lemma}

\begin{proof}
	First, we translate \eqref{equ_model_based_po} into the normal gradient descent form. Let $\psi = M^{-\frac{1}{2}}\theta$. Define $\tilde{J}(\psi) \equiv J(\theta)$ and $\cS_\psi(a) = \{\psi|M^\frac{1}{2}\psi \in \cS_\theta(a)\}$, then 
\bee \label{pg_psi}
\psi^{t+1} = \psi^t - \eta M^{\frac{1}{2}}\nabla_\theta J(\theta^t) = \psi^t - \eta \nabla_\psi \tilde{J}(\psi^t).
\ene
By Lemma \ref{grad_dom} and \ref{lip_smooth}, for any $a \geq J^*$, we have the gradient domination and Lipschitz smoothness property of $\tilde{J}(\psi)$ on $\psi$ as
\bee \label{gd_psi}
\begin{aligned}
&~~~\Vert \nabla_\psi \tilde{J}(\psi)\Vert^2 = \Vert M^{\frac{1}{2}} \nabla_\theta J(\theta) \Vert^2 \geq \underline{\sigma}(M)\Vert\nabla_\theta J(\theta) \Vert^2\\
& \geq \frac{\underline{\sigma}(M)}{\lambda(a)}(J(\theta)-J^*) \geq \frac{\underline{\sigma}(M)}{\lambda(a)}(\tilde{J}(\psi)-J^*),
\end{aligned}
\ene
and
\bee \label{ls_psi}
\begin{aligned}
	&~~~\Vert \nabla_\psi \tilde{J}(\psi) - \nabla_\psi \tilde{J}(\psi')\Vert = \Vert M^{\frac{1}{2}} \nabla_\theta J(\theta) - M^{\frac{1}{2}} \nabla_\theta J(\theta') \Vert\\ &\leq \ell(a)\Vert M \Vert^\frac{1}{2} \Vert \theta - \theta'\Vert \leq  \ell(a)\Vert M \Vert \Vert \psi - \psi'\Vert,
\end{aligned}
\ene
where $\theta$ and $\theta'$ satisfy the requirements in Lemma \ref{grad_dom} and \ref{lip_smooth}.

Then, we show that with stepsize $\eta \in (0,1/(\ell(a) \Vert M \Vert)]$, if $\psi^t \in \cS_\psi(a)$ for some $a>J^*$, then $\psi^{t+1}$ is in $\cS_\psi(a)$ by \eqref{pg_psi}. For simplicity, we use $\psi_{(\eta)}$ to denote $\psi^t - \eta\nabla_\psi \tilde{J}(\psi^t)$ here.

By the smoothness \eqref{ls_psi}, given $\phi \in (0,1)$, there exists $b > 0$ such that $\ell(a+b) \leq (1+\phi) \ell(a)$. Let $\cS_\psi^c(a+b)$ as the complementary set of $\cS_\psi(a+b)$. Clearly, the distance $d = \text{inf}\left\{\Vert\psi - \psi'\Vert\big|\forall \psi \in \cS_\psi(a) \text{ and } \psi' \in \cS_\psi^c(a+b)\right\}$ is larger than $0$. Choose a large enough $\overline{N}$ such that $2/\left(\overline{N}(1+\phi)\ell(a)\Vert M \Vert\right) < d/\Vert \nabla_\psi \tilde{J}(\psi^t) \Vert$, where $\Vert \nabla_\psi \tilde{J}(\psi^t) \Vert > 0$ is ensured by \eqref{gd_psi} and $a>J^*$, and choose $\tau \in \left(0,2/(\overline{N}(1+\phi)\ell(a)\Vert M \Vert)\right]$. Then, it holds that $\Vert \psi_{(\tau)} - \psi^t \Vert < d$, which implies $\psi_{(\tau)} \in \cS_\psi(a+b)$. By the smoothness \eqref{ls_psi} on $\cS_\psi(a+b)$, we have $\tilde{J}(\psi_{(\tau)}) - \tilde{J}(\psi^t) \leq -\tau(1-\tau\ell(a+b)\Vert M \Vert/2)\Vert \nabla_\psi \tilde{J}(\psi^t) \Vert^2\leq0$, which implies that $\psi_{(\tau)} \in \cS_\psi(a)$. Similarly, it holds that $\Vert \psi_{(2\tau)} - \psi_{(\tau)} \Vert < d$, which implies $\psi_{(2\tau)} \in \cS_\psi(a+b)$. By induction we can show that $\psi_{(N\tau)} \in \cS_\psi(a)$ for some $N\in \bN_+$ as long as $\tau N \ell(a+b)\Vert M \Vert/2 < 1$. Since $\ell(a+b) \leq (1+\phi)\ell(a)<2\ell(a)$, for $\eta$ in $(0,1/(\ell(a) \Vert M \Vert)]$, we can choose $\tau$, $N$ such that $\eta = N\tau$ to show $\psi^{t+1} = \psi_{(\eta)} \in \cS_\psi(a)$.

Next, we show that \eqref{pg_psi} converges linearly to $\psi^* = M^{-\frac{1}{2}}\theta^*$. For the stepsize $\eta \in (0,1/(\ell^0 \Vert M \Vert)]$, the update \eqref{pg_psi} satisfies $\psi^t \in \cS_\psi(J(\theta^0)), \forall t\in\bN$. And the cost satisfies 
$$\tilde{J}(\psi^{t+1}) \leq \tilde{J}(\psi^t) - \eta\left(1-\frac{\ell^0\Vert M \Vert \eta}{2}\right)\Vert\nabla_\psi \tilde{J}(\psi^t)\Vert^2.$$
Using the gradient domination property \eqref{gd_psi} we have
$$\tilde{J}(\psi^{t+1}) \leq \tilde{J}(\psi^t) - \nu(\tilde{J}(\psi^t)-J^*).$$
By reorganizing the terms we obtain 
\bee \label{lc_psi}
\tilde{J}(\psi^{t+1}) - J^* \leq \left(1 -\nu\right)(\tilde{J}(\psi^t)-J^*).
\ene

Finally, since $\tilde{J}(\psi) \equiv J(\theta)$, \eqref{lc_psi} is equivalent to 
$$
J(\theta^{t+1}) - J^* \leq \left(1 -\nu\right)(J(\theta^t)-J^*).
$$
\end{proof}

Lemma \ref{theo_conv_theta} shows that the update \eqref{equ_model_based_po} converges linearly to the optimal policy $[\hat{K}~~\hat{l}]$ of \eqref{ce_problem}. As a direct corollary of Lemma \ref{theo_equiv} and \ref{theo_conv_theta}, the convergence of DeePO \eqref{projected_grad} is following.
\begin{theorem} \label{theo_main}
	For any $\xi^0 \in \cS$ and stepsize $\eta \in (0,1/(\ell^0 \Vert M \Vert)]$, the DeePO iteration \eqref{projected_grad} converges linearly:
	$$
	J(\xi^{t+1}) - J^* \leq \left(1 -\nu\right)(J(\xi^t)-J^*), t \in \bN,
	$$
	and the optimal solution is:
	 $$\xi^* = D_0^{-1}\begin{bmatrix}
			\hat{K} & \hat{l} \\
			I_n & 0
		\end{bmatrix}. $$
\end{theorem}

After obtaining the optimal policy $\xi^*$ by \eqref{projected_grad}, the model-based policy $\hat{K},\hat{l}$ can be recovered by \eqref{parametrization}. The proof presented in this section extends and improves upon the sublinear convergence result established in \cite{zhao2024data}. By disregarding the feedforward term in the LQT problem, the results in Lemma \ref{theo_equiv} and the subsequent analysis naturally apply to the DeePO method for the LQR problem as well.

\section{Numerical experiment}
We use a numerical experiment to validate the linear convergence of our method \eqref{projected_grad}. 

We randomly generate the system matrix $(A,B)$ with $n=4,m=2$ from a standard normal distribution and normalize $A$ such that $\rho(A) = 0.8$. In this way, $(A,B)$ is a open-loop stable system, which is convenient for choosing the initial policy. The specific matrices used are given by
$$\begin{aligned}
A&=\begin{bmatrix}
	-0.229&	0.247	&-0.511&	0.493\\
	0.846&	0.159	&0.722&	0.529\\
	-0.018&	0.070	&0.300&	0.758\\
	0.247&	0.546	&-0.511&	-0.176
\end{bmatrix},\\
B&=\begin{bmatrix}
	-0.631&	0.938\\
	0.262&	-0.796\\
	0.461&	-0.180\\
	0.774&	0.112
\end{bmatrix}.
\end{aligned}
$$
It is straightforward to check that $(A,B)$ is controllable. Let $Q = I_4$ and $R = I_2$. The data matrices $(X_0, U_0, X_1)$ are generated with a length of $T = 10$. Here, $x_0$ and $U_0$ are initialized from a normal distribution, and the remaining states are computed using the system dynamics $(A,B)$. We choose $K = 0$ and $l = 0$ as the initial policy, which corresponds to $\xi^0 = 0$. It is first checked that $\rho(\ox_1V) <1$ holds, ensuring $\xi^0\in\cS$. The stepsize is chosen to be $\eta = 0.02$.

\begin{figure}[t!]
	\centering
	\includegraphics[scale=0.6]{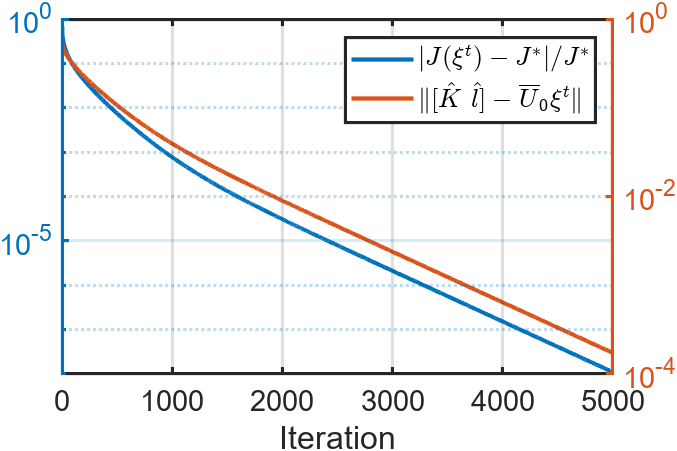}
	\caption{Convergence of DeePO for LQT.}
	\label{fig:con}
\end{figure}

We run the iteration \eqref{projected_grad} for $5000$ times and the results are illustrated in Fig \ref{fig:con}. The blue line shows the error of cost function $J(\xi^t)$, while the red line shows the error between DeePO policy $\xi^t$ and the optimal indirect policy $[\hat{K} ~ \hat{l}]$. Clearly, the results validate Theorem \ref{theo_main}, demonstrating that $\xi^t$ converges linearly to the optimal policy $\xi^*$.

\section{conclusion}
In this work, we extended the DeePO framework to the LQT problem. By introducing a covariance parameterization of the LQT policy, we enabled the policy gradient to be computed directly from data. Subsequently, we proposed a projected policy gradient method—DeePO iteration—to solve the LQT problem. By establishing a connection between DeePO and model-based policy optimization, we proved that the DeePO iteration achieves linear convergence and that the optimal policy is equivalent to the one derived using indirect methods. A numerical experiment was conducted to validate the theoretical findings.

In future work, we will extend the DeePO framework to develop an online version for LQT, building on the results presented in this paper. Additionally, it will be valuable to address the more general case where the reference trajectory varies over time.

\bibliographystyle{IEEEtran}
\bibliography{mybib}

\begin{appendices}
\section{Proofs}
\subsection{Proof of Lemma \ref{theo_cost}}
	To calculate the cost function, define the advantage function
	$f(x;\xi) = \bE \{\sum_{t=0}^{\infty}(x_t-\delta)^{\top}Q(x_t-\delta)+u_t^{\top}Ru_t - J(\xi)|x_0=x, u_t = \ou_0Vx_t + \ou_0h\}.$
	By using backward DP \cite{bertsekas1995dynamic}, it can be shown that $f(x;V,h)$ has a quadratic form $f(x;V,h) = x^{\top}Px + 2g^{\top}x+s$, where $P, g \text{ and } s$ are to be determined. By the Bellman equation, it holds that $f(x_t;\xi) = \bE \{(x_t-\delta)^{\top}Q(x_t-\delta)+u_t^{\top}Ru_t - J(\xi) + f(x_{t+1};\xi) \}$. Then we have
	\been
	\begin{aligned}
		P &= Q+V^\top \ou_0^\top R \ou_0V+V^\top \ox_1^\top P\ox_1V,\\
		g^\top &= -\delta^\top Q + h^\top \ou_0^\top R\ou_0V + h^\top \ox_1^\top P \ox_1 V + g^\top \ox_1V,\\
		s &= \delta^\top Q \delta + h^\top \ou_0^\top R\ou_0 h + h^\top \ox_1^\top P\ox_1 h + 2g^\top \ox_1h \\
		&~~+ \tr(PW) -J(\xi) + s.
	\end{aligned}
	\enen
	
	Therefore, we obtain that $P = P_V$, $g = g_\xi$, and the proof is complete by the third equation.

\subsection{Proof of Lemma \ref{theo_gradient}}
	First, we give the gradient with respect to $h$:
	\bee \label{grad_h}
	\begin{aligned} 
		\nabla_h J(V,h) &= 2(\ou_0^\top R\ou_0 h + \ox_1^\top P_V\ox_1 h + \ou_0^\top R\ou_0VY_V\ox_1h \\
		&~~~+ \ox_1^\top Y_V^\top V^\top\ou_0^\top R\ou_0 h + \ox_1^\top P_V \ox_1 VY_V\ox_1h \\
		&~~~+ \ox_1^\top Y_V^\top V^\top \ox_1^\top P_V \ox_1 h -\ox_1^\top Y_V^\top Q\delta) \\
		&= 2(\ox_1^\top g_\xi + E_V\overline{x} + \ou_0^\top R\ou_0 h\ + \ox_1^\top P_V\ox_1 h)\\
		&= 2(G_\xi + E_V\overline{x}).
	\end{aligned}
	\ene
	
	Using the stationary state distribution $\Gamma$, the cost $J(\xi)$ can be written equivalently as 
	\been
	\begin{aligned}
		J(\xi) &= \bE_{\overline{x}\sim\Gamma}\left\{(\overline{x}-\delta)^\top Q (\overline{x}-\delta) \right.\\
		&~~~\left.+ (\ou_0V\overline{x}+\ou_0h)^\top R(\ou_0V\overline{x}+\ou_0h)\right\}\\
		&= \tr\left\{(Q+V^\top\ou_0^\top R \ou_0 V)(\Sigma_V + \overline{x}\overline{x}^T)\right\} + \\
		&~~~2(-\delta^\top Q+h^\top \ou_0^\top R \ou_0 V)\overline{x} + \delta^\top Q \delta + h^\top \ou_0^\top R\ou_0 h.
	\end{aligned}
	\enen
	Then, the gradient with respect to $V$ is given by
	\been
	\begin{aligned}
		\nabla_V J(V,h) &= 2E_V\Sigma_V + 2\ox_1^\top Y_V^\top Q \overline{x}\overline{x}^\top + 2\ou_0^\top R \ou_0 V \overline{x} \overline{x}^\top \\
		&~~~+ 2\ox_1^\top Y_V^\top V^\top \ou_0^\top R \ou_0 V \overline{x} \overline{x}^\top -2\ox_1^\top Y_V^\top Q\delta\overline{x}^\top\\
		&~~~ + 2\ou_0^\top R\ou_0 h\overline{x}^\top + 2\ox_1^\top Y_V^\top V^\top \ou_0^\top R\ou_0 h\overline{x}^\top.
	\end{aligned}
	\enen
	Also, the gradient with respect to $h$ can be written equivalently as 
	\been
	\begin{aligned}
		\nabla_h J(V,h) &= 2\ox_1^\top Y_V^\top (Q+V^\top\ou_0^\top R \ou_0 V) \overline{x} + 2\ou_0^\top R\ou_0 h\\
		&~~~ -2\ox_1^\top Y_V^\top Q\delta +2\ou_0^\top R \ou_0 V \overline{x}\\
		&~~~ +2\ox_1^\top Y_V^\top V^\top \ou_0^\top R\ou_0 h.
	\end{aligned}
	\enen
	
	We can find that $\nabla_V J(V,h) = 2E_V\Sigma_V + \nabla_h J(V,h)\overline{x}^\top$. Combined with \eqref{grad_h}, we have
	\been
		\nabla_\xi J(\xi) =2\left[E_V~G_\xi\right] \Phi_{\xi}.
	\enen
	
\subsection{Proof of Lemma \ref{theo_equiv}}
	Based on \eqref{conn_data_estimation}, \eqref{conn_data_estimation_2}, \eqref{parametrization} and \eqref{projected_grad}, we have
	$$
	\begin{aligned} \label{iter_K_deepo}
		&~~~~K^+ - K \\
		&=2\eta\ou_0(I-\ox_0^\dagger \ox_0)(E_V\Sigma_V + E_V\overline{x}\overline{x}^\top+G_\xi\overline{x}^\top) ~~~~~~~\\
	\end{aligned}
	$$
	\bee
	\begin{aligned}
		&=2\eta\ou_0(I-\ox_0^\dagger \ox_0)\left((\ou_0^\top R\ou_0V + \ox_1^\top P_V\ox_1V)(\Sigma_V+\overline{x}\overline{x}^\top) \right.\\
		& ~~~\left.+(\ox_1^\top g_\xi + \ou_0^\top R\ou_0 h\ + \ox_1^\top P_V\ox_1 h)\overline{x}^\top\right) \\
		&=2\eta M\left(( R\ou_0V + \hb^\top P_V\ox_1V)(\Sigma_V+\overline{x}\overline{x}^\top)\right.\\
		&~~~\left.+(\hb^\top g_\xi + R\ou_0 h\ + \hb^\top P_V\ox_1 h)\overline{x}^\top\right),\\
	\end{aligned}
	\ene
	where the last equation is because \eqref{conn_data_estimation} and $(I-\ox_0^\dagger \ox_0)\ox_0^\top=0$.
	
	According to \cite{zhao2023global}, the gradient of $J(\theta)$ in $K$ is 
	\bee
	\begin{aligned} \label{grad_K}
		\nabla_K J(\theta) &= 2\left(( RK + \hb^\top P_V(\ha+\hb K))(\Sigma_V+\overline{x}\overline{x}^\top)\right.\\
		&~~~ \left.+ (\hb^\top g_\xi + RK + \hb^\top P_V(\ha+\hb K))\overline{x}^\top\right) \\
		&=2\left(( R\ou_0V + \hb^\top P_V\ox_1V)(\Sigma_V+\overline{x}\overline{x}^\top)\right.\\
		&~~~ \left.+(\hb^\top g_\xi + R\ou_0 h + \hb^\top P_V\ox_1 h)\overline{x}^\top)\right).
	\end{aligned}
	\ene
	By \eqref{iter_K_deepo} and \eqref{grad_K} we have 
	\bee \label{equvalence_K}
	K^+ - K = -\eta M\nabla_K J(\xi). 
	\ene
	
	Similarly, it also holds 
	\bee
	\begin{aligned}
		l^+ - l &=2\eta\ou_0(I-\ox_0^\dagger \ox_0)(E_V\overline{x}+G_\xi) \\
		&=2\eta\ou_0(I-\ox_0^\dagger \ox_0)\left((\ou_0^\top R\ou_0V + \ox_1^\top P_V\ox_1V)\overline{x} \right.\\
		&~~~\left.+\ox_1^\top g_\xi + \ou_0^\top R\ou_0 h\ + \ox_1^\top P_V\ox_1 h\right) \\
		&=2\eta M\left((R\ou_0V + \hb^\top P_V\ox_1V)\overline{x}\right.\\
		&~~~ \left.+ \hb^\top g_\xi + R\ou_0 h\ + \hb^\top P_v\ox_1 h\right).
	\end{aligned}
	\ene
	According to \cite{zhao2023global}, the gradient of $J(\theta)$ in $l$ is 
	\bee
	\begin{aligned}
		\nabla_l J(\theta) &= 2\left((R\ou_0V + \hb^\top P_V\ox_1V)\overline{x} \right.\\
		&~~~\left.+\hb^\top g_\xi + R\ou_0 h\ + \hb^\top P_V\ox_1 h\right).
	\end{aligned}
	\ene
	Therefore we have 
	\bee \label{equvalence_l}
	l^+ - l = -\eta M\nabla_l J(\xi). 
	\ene
	
	By combining \eqref{equvalence_K} and \eqref{equvalence_l}, the projected gradient method on $\xi$ \eqref{projected_grad} is equivalent to the following gradient decent method on $\theta = [K,l]$:
	\bee
	\theta^+ = \theta - \eta M \nabla_\theta J(\theta).
	\ene
	
	Next, we show that $M$ is positive definite. Let $N = \ou_0(I-\ox_0^\dagger \ox_0)$. It holds that 
	\bee \label{M_equ_ntn}
	M = N N^\top + \ou_0\ox_0^\dagger\ox_0(I-\ox_0^\dagger \ox_0)\ou_0^\top = N N^\top.
	\ene
	By the persistent excitation assumption, we have
	$$\text{rank}(\Lambda) = 
	\text{rank}\begin{bmatrix}
		N + \ou_0\ox_0^\dagger\ox_0 \\ \ox_0
	\end{bmatrix} = \text{rank}\begin{bmatrix}
		N \\ \ox_0
	\end{bmatrix}=n+m.
	$$
	
	Then, we have $\text{rank}(N) = m$. Together with \eqref{M_equ_ntn}, it holds that $M$ is positive definite.

\end{appendices}
\end{document}